\documentclass[a4paper,UKenglish]{lipics}
\pdfoutput=1
\usepackage{microtype}
\usepackage[final]{hyll}

\title{A Hybrid Linear Logic for Constrained Transition Systems}

\titlerunning{A Hybrid Linear Logic for Constrained Transition Systems}

\author[1]{Jo\"elle Despeyroux}
\author[2]{Kaustuv Chaudhuri}

\affil[1]{INRIA and CNRS, I3S, Sophia-Antipolis, France\\
  \texttt{joelle.despeyroux@inria.fr}}
\affil[2]{INRIA, France\\\texttt{kaustuv.chaudhuri@inria.fr}}

\authorrunning{J. Despeyroux and K. Chaudhuri}

\Copyright{Jo\"elle Despeyroux and Kaustuv Chaudhuri} %% mandatory.



\serieslogo{}%
\volumeinfo%
       {Ralph Matthes and Aleksy Schubert}%
       {2}%
       {19th International Conference on Types for Proofs and Programs (TYPES~2013)}%
       {27}% volume number
       {1}%
       {151}%  first page of the article (will be set by the editors)
\EventShortName{TYPES 2013}
\DOI{10.4230/LIPIcs.TYPES.2013.XXX}% XXX is not the final value!

\begin{document}

\maketitle

\begin{abstract}
  Linear implication can represent state transitions, but real transition
  systems operate under temporal, stochastic or probabilistic constraints that
  are not directly representable in ordinary linear logic.
  We propose a general modal extension of intuitionistic linear logic where
  logical truth is indexed by constraints and hybrid connectives combine
  constraint reasoning with logical reasoning.%
  The logic has a focused cut-free sequent calculus that can be used to
  internalize the rules of particular constrained transition systems; we
  illustrate this with an adequate encoding of the synchronous stochastic
  pi-calculus.
\end{abstract}

\section{Introduction}

To reason about state transition systems, we need a logic of state.
Linear logic~\cite{girard87tcs} is such a logic and has been successfully used
to model such diverse systems as process calculi~\cite{miller92welp}, references
and concurrency in programming languages~\cite{wadler90ifiptc2}, and formal
security~\cite{caires10concur,cervesato00entcs}, to give a few examples.
Linear logic achieves this versatility by representing propositions as
\emph{resources} that are combined using "tens", which can then be transformed
using the linear implication ("-o").
However, linear implication is timeless: there is no way to correlate two
concurrent transitions.
If resources have lifetimes and state changes have temporal, probabilistic or
stochastic \emph{constraints}, then the logic will allow inferences that may not
be realizable in the system being modelled.
The need for formal reasoning in such constrained systems has led to the
creation of specialized formalisms such as Computation Tree Logic
(\proofsystem{CTL})\cite{Emerson95}, Continuous Stochastic Logic
(\proofsystem{CSL})~\cite{aziz00tcl} or Probabilistic CTL
(\proofsystem{PCTL})~\cite{hansson94fac}.
These approaches pay a considerable encoding overhead for the states and
transitions in exchange for the constraint reasoning not provided by linear logic.
A prominent alternative to the logical approach is to use a suitably enriched
process algebra such as the stochastic and probabilistic $\pi$-calculi or the
$\kappa$-calculus~\cite{danos04tcs}.
Processes are animated by means of simulation and then compared with the
observations.
Process calculi do not however completely fill the need for {\em formal
  reasoning for constrained transition systems}.

We propose a simple yet general method to add constraint reasoning to linear logic.
It is an old idea---\emph{labelled deduction}~\cite{simpson94phd} with
\emph{hybrid} connectives~\cite{brauener06jal}---applied to a new domain.
Precisely, we parameterize ordinary logical truth on a \emph{constraint domain}:
"A @ w" stands for the truth of "A" under constraint "w".
Only a basic monoidal structure is assumed about the constraints from a
proof-theoretic standpoint.
We then use the hybrid connectives of \emph{satisfaction} ("at") and
\emph{localization} ("dn") to perform generic symbolic reasoning on the
constraints at the propositional level.
We call the result \emph{hybrid linear logic} (\hyll); it has a generic cut-free
(but cut admitting) sequent calculus that can be strengthened with a focusing
restriction~\cite{andreoli92jlc} to obtain a normal form for proofs.
Any instance of \hyll that gives a semantic interpretation to the constraints
enjoys these proof-theoretic properties.

Focusing allows us to treat \hyll as a \emph{logical framework} for constrained
transition systems.
Logical frameworks with hybrid connectives have been considered before; hybrid
\LF (\HLF), for example, is a generic mechanism to add many different kinds of
resource-awareness, including linearity, to ordinary \LF~\cite{reed06hylo}.
\HLF follows the usual \LF methodology of keeping the logic of the framework
minimal: its proof objects are $\beta$-normal $\eta$-long natural deduction
terms, but the equational theory of such terms is sensitive to permutative
equivalences~\cite{watkins03tr}.
With a focused sequent calculus, we have more direct access to a canonical
representation of proofs, so we can enrich the framework with any connectives
that obey the focusing discipline.
The representational adequacy of an encoding in terms of (partial) focused
sequent derivations tends to be more straightforward than in a natural deduction
formulation.
We illustrate this by encoding the synchronous stochastic $\pi$-calculus (\Spi)
in \hyll using rate functions as constraints.

In addition to the novel stochastic component, our encoding of \Spi is a
conceptual improvement over other encodings of $\pi$-calculi in linear logic.
In particular, we perform a full propositional reflection of processes as
in~\cite{miller92welp}, but our encoding is first-order and adequate as
in~\cite{cervesato03tr}.
\hyll does not itself prescribe an operational semantics for the encoding of
processes; thus, bisimilarity in continuous time Markov chains (\CTMC) is not
the same as logical equivalence in stochastic \hyll, unlike in
\proofsystem{CSL}~\cite{desharmais03jlap}.
This is not a deficiency; rather, the \emph{combination} of focused \hyll proofs
and a proof search strategy tailored to a particular encoding is necessary to
produce faithful symbolic executions.
This exactly mirrors \Spi where it is the simulation rather than the transitions
in the process calculus that is shown to be faithful to the \CTMC
semantics~\cite{phillips04cmmb}.

This work has the following main contributions.
First is the logic \hyll itself and its associated proof-theory, which has a
%% clean judgemental pedigree in the Martin-Löf tradition.
very standard and well understood design.
%% *** in the report:
%% in the Martin-Löf tradition.
%
Second, we show how to obtain many different instances of \hyll for particular
constraint domains because we only assume a basic monoidal structure for
constraints.
Third, we illustrate the use of focused sequent derivations to obtain adequate
encodings by giving a novel adequate encoding of \Spi.
Our encoding is, in fact, \emph{fully adequate}, \ie, partial focused proofs are
in bijection with traces.
The ability to encode \Spi gives an indication of the versatility of \hyll.

This paper is organized as follows: in \secref{hyll}, we present the inference
(sequent calculus) rules for \hyll and describe the two main semantic instances:
temporal and probabilistic constraints.
In \secref{focusing} we sketch the general focusing restriction on \hyll sequent
proofs.
In \secref{spi} we give the encoding of \spi in probabilistic \hyll, and show that
the encoding is representationally adequate for focused proofs (theorems
\ref{thm:completeness} and \ref{thm:adeq}).
We end with an overview of related (\secref{related}) and future work
(\secref{concl}).
The full version of this paper is available as a technical
report~\cite{Chaudhuri-Despeyroux:13tr}.

\section{Hybrid Linear Logic}
\label{sec:hyll}

In this section we define \hyll, a conservative extension of intuitionistic
first-order linear logic (\ill)~\cite{girard87tcs} where the truth judgements
are labelled by worlds representing constraints. Like in \ill, propositions are
interpreted as \emph{resources} which may be composed into a \emph{state} using
the usual linear connectives, and the linear implication ("-o") denotes a
transition between states. The world label of a judgement represents a
constraint on states and state transitions; particular choices for the worlds
produce particular instances of \hyll. The common component in all the instances
of \hyll is the proof theory, which we fix once and for all. We impose the
following minimal requirement on the kinds of constraints that \hyll can deal
with.

\begin{definition} \label{defn:constraint-domain}
  A \emph{constraint domain} "\cal W" is a monoid structure
  "\langle W, ., rid\rangle". The elements of "W" are called \emph{worlds}, and
  the partial order "\preceq\ : W \times W"---defined as "u \preceq w" if there
  exists "v \in W" such that "u . v = w"---is the \emph{reachability relation}
  in "\cal W".
\end{definition}

\noindent
The identity world "rid" is "\preceq"-initial and is intended to represent the
lack of any constraints. Thus, the ordinary \ill is embeddable into any instance
of \hyll by setting all world labels to the identity. When needed to
disambiguate, the instance of \hyll for the constraint domain "\cal W" will be
written \chyll[W].
Two design choices are important to note.
First, we only require the worlds to be monoids, not lattices, because this lets
us give a sufficiently general system that it can be instantiated with rate
functions as the constraint domain.
Second, we do not assume that the monoid is commutative so that we can still
choose to use lattices for the constraint domain.

Atomic propositions are written using lowercase letters ("a, b, ...") applied to
a sequence of \emph{terms} ("s, t, \ldots"), which are drawn from an untyped
term language containing term variables ("x, y, \ldots") and function symbols
("f, g, ...") applied to a list of terms. Non-atomic propositions are
constructed from the connectives of first-order intuitionistic linear logic and
the two hybrid connectives \emph{satisfaction} ("at"), which states that a
proposition is true at a given world ("w, u, v, \ldots"), and
\emph{localization} ("dn"), which binds a name for the (current) world the proposition is
true at. The following grammar summarizes the syntax of \hyll propositions.

\smallskip
\bgroup
\begin{tabular}{l@{\ }r@{\ }l}
  "A, B, ..." & "::=" & "a~\vec t\ OR A tens B OR one OR A -o B OR A with B OR top OR A plus B OR zero OR ! A OR all x. A OR ex x. A" \\
              & "|\ " & "(A at w) OR now u. A OR all u. A OR ex u. A" \\
\end{tabular}
\egroup

\smallskip
\noindent
Note that in the propositions "now u. A", "all u. A" and "ex u. A",
%% WAS: the scope of the world variable "u" is all the worlds occurring in "A".
world $u$ is bound in $A$.
World variables cannot be used in terms, and neither can term variables occur in worlds;
this restriction is important for the modular design of \hyll because it keeps purely
logical truth separate from constraint truth.  We let "\alpha" range over variables of either kind.
%
%% WAS: Note that
As we shall prove later (Theorem \ref{thm:invertibility}),
the "dn" connective commutes with every propositional connective, including itself. That
is, "now u. (A * B)" is equivalent to "(now u. A) * (now u. B)" for all binary connectives "*", and
"now u. * A" is equivalent to "* (now u. A)" for every unary connective "*", assuming the
commutation will not cause an unsound capture of "u". It is purely a matter of taste where to place
the "dn", and repetitions are harmless.

The unrestricted connectives "and", "or", "imp", \etc of intuitionistic
(non-linear) logic can also be defined in terms of the linear connectives and
the exponential "!"  using any of the available embeddings of intuitionistic
logic into linear logic, such as Girard's embedding~\cite{girard87tcs}.
For example,   "A imp B" can be defined as "! A -o B".

\subsection{Sequent Calculus for \hyll}

In this section, we give a sequent calculus presentation of \hyll and prove a
cut-admissibility theorem.  The sequent formulation in turn will lead to an
analysis of the polarities of the connectives in order to get a focused sequent
calculus that can be used to compile a logical theory into a system of derived
inference rules with nice properties (\secref{focusing}).  For instance, if a
given theory defines a transition system, then the derived rules of the focused
calculus will exactly exhibit the same transitions. This is key to obtain the
necessary representational adequacy theorems,
as we shall see for the \spi-calculus example chosen in this paper (\secref{spi.adq}).

We start with the judgements from linear logic~\cite{girard87tcs} and enrich
them with a modal situated truth. We present the syntax of hybrid linear logic
in a sequent calculus style, using Martin-L\"{o}f's principle of separating
judgements and logical connectives.
Instead of the ordinary mathematical judgement ``"A" is true'', for a proposition $A$,
judgements of \hyll are of the form ``"A" is true at world "w"'', abbreviated as "A @ w".
We use sequents of
the form "\G ; \D \ ==> \ C @ w" where "\G" and "\D" are sets of judgements of
the form "A @ w", with "\D" being moreover a \emph{multiset}.  "\G" is called
the \emph{unrestricted context}: its hypotheses can be consumed any number of
times.  "\D" is a \emph{linear context}: every hypothesis in it must be consumed
singly in the proof.
Note that in a judgement "A @ w" (as in a proposition "A at w"), $w$ can be
any expression in $\cal W$, not only a variable.
The notation $ A [ \tau / \alpha ]$ stands for the replacement of all free occurrences
of the variable $\alpha$ in $A$ with the expression $\tau$, avoiding capture.
Note that the expressions in the rules are to be read up to alpha-conversion.

\begin{figure}[p]
\framebox{
\begin{minipage} {0.9\textwidth}
  \setlength{\parindent}{0pt}
   \bgroup \small
%%   \paragraph{Judgemental rules}
   {\bf \large Judgemental rules}
   \begin{gather*}
     \I[init]{"\G ; a\ \vec t\ @ u ==> a\ \vec t\ @ u"}{}
     \SP
     \I[copy]{"\G, A @ u ; \D ==> C @ w"}{"\G, A @ u ; \D, A @ u ==> C @ w"}
   \end{gather*}

%%   \paragraph{Multiplicatives}
   {\bf \large Multiplicatives}
   \begin{gather*}
     \I["{tens}R"]{"\G ; \D, \D' ==> A tens B @ w"}{"\G ; \D ==> A @ w" & "\G ; \D' ==> B @ w"}
     \SP
     \I["{tens}L"]{"\G ; \D, A tens B @ u ==> C @ w"}
     {"\G ; \D, A @ u, B @ u ==> C @ w"}
     \\[1ex]
     \I["{one}R"]{"\G ; . ==> one @ w"}
     \SP
     \I["{one}L"]{"\G ; \D, one @ u ==> C @ w"}{"\G ; \D ==> C @ w"}
     \\[1ex]
     \I["{-o}R"]{"\G ; \D ==> A -o B @ w"}{"\G ; \D, A @ w ==> B @ w"}
     \SP
     \I["{-o}L"]{"\G ; \D, \D', A -o B @ u ==> C @ w"}
     {"\G ; \D ==> A @ u" & "\G ; \D', B @ u ==> C @ w"}
   \end{gather*}

%%   \paragraph{Additives}
   {\bf \large Additives}
   \begin{gather*}
     \I["top R"]{"\G ; \D ==> top @ w"}
     \LSP
     \I["zero L"]{"\G ; \D, zero @ u ==> C @ w"}
     \\[1ex]
     \I["{with}R"]{"\G ; \D ==> A with B @ w"}{"\G ; \D ==> A @ w" & "\G ; \D ==> B @ w"}
     \LSP
     \I["{with}L_i"]{"\G ; \D, A_1 with A_2 @ u ==> C @ w"}
     {"\G ; \D, A_i @ u ==> C @ w"}
     \\[1ex]
     \I["{plus}R_i"]{"\G ; \D ==> A_1 plus A_2 @ w"}{"\G ; \D ==> A_i @ w"}
     \LSP
     \I["{plus}L"]{"\G ; \D, A plus B @ u ==> C @ w"}
     {"\G ; \D, A @ u ==> C @ w" & "\G ; \D, B @ u ==> C @ w"}
   \end{gather*}

%%   \paragraph{Quantifiers}
   {\bf \large Quantifiers}
   \begin{gather*}
     \I["\forall R^\alpha"]{"\G ; \D ==> \fall \alpha A @ w"}{"\G ; \D ==> A @ w"}
     \SP
     \I["\forall L"]{"\G ; \D, \fall \alpha A @ u ==> C @ w"}
     {"\G ; \D, [\tau / \alpha] A @ u ==> C @ w"}
     \\[1ex]
     \I["\exists R"]{"\G ; \D ==> \fex \alpha A @ w"}{"\G ; \D ==> [\tau / \alpha] A @ w"}
     \SP
     \I["\exists L^\alpha"]{"\G ; \D, \fex \alpha A @ u ==> C @ w"}
     {"\G ; \D, A @ u ==> C @ w"}
   \end{gather*}

   For "\forall R^\alpha" and "\exists L^\alpha", "\alpha" is assumed to
   be fresh with respect to the conclusion. For "\exists R" and "\forall
   L", "\tau" stands for a term or world, as appropriate.

%%   \paragraph{Exponentials}
   {\bf \large Exponentials}
   \begin{gather*}
     \I["{!}R"]{"\G ; . ==> {! A} @ w"}{"\G ; . ==> A @ w"}
     \SP
     \I["{!}L"]{"\G ; \D, {! A} @ u ==> C @ w"}{"\G, A @ u ; \D ==> C @ w"}
   \end{gather*}

%%   \paragraph{Hybrid connectives}
   {\bf \large Hybrid connectives}
   \begin{gather*}
     \I["at R"]{"\G ; \D ==> (A at u) @ v"}{"\G ; \D ==> A @ u"}
     \LSP
     \I["at L"]{"\G ; \D, (A at u) @ v ==> C @ w"}{"\G ; \D, A @ u ==> C @ w"}
     \\[1ex]
     \I["{dn}R"]{"\G ; \D ==> now u. A @ w"}{"\G ; \D ==> [w/u] A @ w"}
     \LSP
     \I["{dn}L"]{"\G ; \D, now u. A @ v ==> C @ w"}{"\G ; \D, [v/u] A @ v ==> C @ w"}
   \end{gather*}
   \egroup
\end{minipage}}
\caption{The sequent calculus for \hyll}.
\label{fig:seq-rules}
\end{figure}

The full collection of rules of the \hyll sequent calculus is in
\figref{seq-rules}.
The rules for the linear connectives are borrowed from \cite{chaudhuri03tr}
where they are discussed at length, so we omit a more thorough discussion here.
The rules for the first-order quantifiers are completely standard.
A brief discussion of the hybrid rules follows.
To introduce the \emph{satisfaction}
proposition "(A at u)" (at any world "v'") on the right, the proposition "A" must be true in the
world "u". The proposition "(A at u)" itself is then true at any world, not just
in the world "u". In other words, "(A at u)" carries with it the world at which
it is true. Therefore, suppose we know that "(A at u)" is true (at any world "v");
then, we also know that "A @ u".
The other hybrid connective of \emph{localisation}, "dn", is intended to be
able to name the current world. That is, if "now u. A" is true at world "w",
then the variable "u" stands for "w" in the body "A". This interpretation is
reflected in its introduction rule on the right "{dn} R". For left introduction,
suppose we have a proof of "now u. A @ v" for some world "v".
Then, we also know "[v / u] A @ v".

There are only two structural rules: the init rule infers an
atomic initial sequent, and the copy rule introduces a contracted copy of an
unrestricted assumption into the linear context (reading from conclusion to
premise). Weakening and contraction are admissible rules:

\begin{theorem}[structural properties] \mbox{}
  \begin{ecom}
  \item If "\G ; \D ==> C @ {w}", then "\G, \G' ; \D ==> C @ {w}". (weakening)
  \item If "\G, A @ u, A @ u ; \D ==> C @ {w}", then "\G, A @ u ; \D ==> C @ {w}". (contraction)
  \end{ecom}
\end{theorem}

\begin{proof}
  By straightforward structural induction on the given derivations.
\end{proof}

The most important structural properties are the admissibility of the identity
and the cut principles. The identity theorem is the general case of the init
rule and serves as a global syntactic completeness theorem for the
logic. Dually, the cut theorem below establishes the syntactic soundness of the
calculus; moreover there is no cut-free derivation of ". ; . ==> zero @ w", so
the logic is also globally consistent.

\begin{theorem}[identity]
  "\G ; A @ w ==> A @ w".
\end{theorem}

\begin{proof}
  By induction on the structure of "A" (see \cite{Chaudhuri-Despeyroux:13tr}).
\end{proof}

\bgroup
\begin{theorem}[cut] \mbox{} \label{thm:cut}
  \begin{ecom}[1.]
  \item If "\G ; \D ==> A @ u" and "\G ; \D', A @ u ==> C @ w", then "\G ; \D,
    \D' ==> C @ w".
  \item If "\G ; . ==> A @ u" and "\G, A @ u ; \D ==> C @ w", then "\G ; \D
    ==> C @ w".
  \end{ecom}
\end{theorem}
\egroup

\begin{proof}
  By lexicographic structural induction on the given derivations, with cuts of
  kind 2 additionally allowed to justify cuts of kind 1. The style of proof
  sometimes goes by the name of \emph{structural
    cut-elimination}~\cite{chaudhuri03tr}.
  See \cite{Chaudhuri-Despeyroux:13tr} for the details.
\end{proof}

We can use the admissible cut rules to show that the following rules are
invertible: "tens L", "one L", "plus L", "zero L", "\exists L", "-o R", "with
R", "top R", and "\forall R". In addition, the four hybrid rules, "at R", "at
L", "{dn} R" and "{dn} L" are invertible. In fact, "dn" and "at" commute freely
with all non-hybrid connectives:

\begin{theorem}[Invertibility] \mbox{} \label{thm:invertibility}
The following rules are invertible:
  \begin{ecom}
  \item On the right: "with R", "top R", "{-o} R", "\forall R", "{dn} R" and "at R";
  \item On the left: "tens L", "one L", "plus L", "zero L", "\exists L", "!L", "{dn} L" and "at L". \qed
  \end{ecom}
\end{theorem}

\begin{corollary}[consistency]
  There is no proof of ". ; . ==> zero @ w".
\end{corollary}

\begin{proof}
A straightforward consequence of \thmref{cut}.
  % Suppose ". ; . ==> zero @ w" is derivable. Then, by cut-admissibility
  % theorems on the sequent calculus, ". ; . ==> zero @ w" must have a cut-free proof.
  % But, we can see by simple inspection that there can be no cut-free proof of
  % ". ; . ==> zero @ w", as this sequent cannot be the conclusion of any rule of inference
  % in the sequent calculus.
  % Therefore, ". ; . ==> zero @ w" is not derivable.
\end{proof}

\hyll is conservative with respect to ordinary intuitionistic linear logic: as long as
no hybrid connectives are used, the proofs in \hyll are identical to those in
\ill~\cite{chaudhuri03tr}. The proof (omitted) is by simple structural
induction.

\begin{theorem}[conservativity]
  Call a proposition or multiset of propositions \emph{pure} if it contains no
  instance of the hybrid connectives and no instance of quantification over a world variable,
  and let "\G", "\D" and "A" be pure.
  Then,
  If "\G ; \D ==>_{\hyll} C @ w" is derivable, then so is "\G ; \D ==>_{\ill} C". \qed
\end{theorem}

% An example of derived statements, true in every semantics for worlds, is the following:
% %
% \begin{proposition}[relocalisation]
% \label{thm:relocalisation}
% \begin{gather*}
%   \Ic[]{"\G ; A_1 @ {u . w_1} \cdots A_k @ {u . w_k} |- B @ {u . v}"}
%          {"\G ; A_1 @ w_1 \cdots  A_k @ w_k |- B @ v"}
% \end{gather*}
% \end{proposition}
% %
% This property is particularly well suited to applications in biology, for example.
% %

%% \subsection{Modal Connectives}
%% \paragraph{Modal Connectives}

In the rest of this paper we use the following derived connectives:

\bgroup %\small
\begin{definition}[modal connectives] \label{defn:connectives} \mbox{}
  \vspace{-1ex}
  \begin{gather*}
    \begin{aligned}
      "box A" &\triangleq "now u. all w. (A at u . w)" & \qquad
      "dia A" &\triangleq "now u. ex w. (A at u . w)" \\
      "rate v A" &\triangleq "now u. (A at u . v)" &
      "!! A" &\triangleq "all u. (A at u)"
    \end{aligned}
  \end{gather*}
\end{definition}
\egroup

\noindent
The connective "\rho" represents a form of delay. Note its derived right rule:
\begin{gather*}
  \Ic["\rho R"]{"\G ; \D |- rate v A @ w"}{"\G ; \D |- A @ {w . v}"}
\end{gather*}
The proposition "rate v A" thus stands for an \emph{intermediate state} in a
transition to "A".
Informally it can be thought to be ``"v" before "A"''; thus, "box A = all v.
rate v A" represents \emph{all} intermediate states in the path to "A", and "dia
A = ex v. rate v A" represents \emph{some} such state.
%% "box A" [resp. "dia A"] represents all [resp. some] state(s) satisfying $A$ and reachable from now.
The modally unrestricted proposition "!! A" represents a resource that is
consumable in any world; it is intended to be used to make the transition rules
available everywhere.

It is worth remarking that the reachability relation in \hyll is trivial: every
world that can be defined is reachable from every other.
To illustrate, the (linear form of the) axioms of the S5 modal logic are
derivable in \hyll; in particular, the sequent ". ; dia A @ w ==> box dia A @
w", which represents the 5 axiom, is provable.
\hyll is, in fact, more expressive than S5 as it allows direct manipulation of
the worlds using the hybrid connectives: for example, the $\rho$ connective is
not definable in S5.

\subsection{Temporal Constraints}
\label{sec:hyllt}

As a pedagogical example, consider the constraint domain "\mathcal{T} = \langle
\Reals^+, +, 0\rangle" representing instants of time. This domain can be used to
define the lifetime of resources, such as keys, sessions, or delegations of
authority. Delay (\defnref{connectives}) in \hyllt represents intervals of time;
"rate d A" means ``"A" will become available after delay "d"'', similar to
metric tense logic~\cite{prior57book}.
This domain is very permissive because addition is commutative, resulting in the
equivalence of "rate u rate v A" and "rate v rate u A".
The ``forward-looking'' connectives
%% "G" and "F" of ordinary tense logic
"G" (always in the future) and "F" (sometimes in the future) of ordinary tense logic
are precisely "box" and "dia" of \defnref{connectives}.

%% \paragraph{}
%% Past connectives:
%% ``Y: yesterday'' dual of X - S: ``since'' dual of U -
%% H: ``historically'' dual of G  - O: ``once'' dual of F
In addition to the
future connectives, the domain "\mathcal{T}" also admits past connectives if we add
saturating subtraction (\ie, "a - b = 0" if "b \ge a") to the language of
worlds. We can then define the duals
%% "H" and "O" of "G" and "F" as:
"H" (historically) and "O" (once) of "G" and "F" as:
%
%% \begin{align*}
%%   "H~A" &\triangleq "now u. all w. (A at u - w)"  \\
%%   "O~A" &\triangleq "now u. ex w. (A at u - w)"
%% \end{align*}
\begin{gather*}
\textrm{H~A} \triangleq\ "now u. all w. (A at u - w)"
\qquad
\textrm{O~A} \triangleq\ "now u. ex w. (A at u - w)"
\end{gather*}
While this domain does not have any branching structure like CTL, it is
expressive enough for many common idioms because of the branching structure
of derivations involving $\oplus$. CTL reachability (``in some path in some
future''), for instance, is the same as our "dia"; similarly
CTL steadiness (``in some path for all futures'') is the same as $\Box$.
%% There is some loss of expressive power, however; for instance, in CTL steadiness
%% (``in some path for all futures'') is distinct from stability, whereas the best
%% approximation in \hyll is "now u. ex w. box (A at u . w)".
CTL stability (``in all paths in all futures''), however, has no direct correspondance in HyLL
%% In the report: %% Only in the report?
%% [
(see however \cite{deMaria-Despeyroux-Felty:14rep}
for a correspondance in particular cases).
%% ]
Note that model checking cannot cope with temporal expresssions involving the
``in all paths'' notion anyway
\footnote{
at least in their full generality, involving an infinite number of states.
%% WAS: or an infinite (not looping) path. %% Implies an infinite number of states.
}.

On the other hand, the availability of linear reasoning, enriched with
modalities, makes certain kinds of reasoning in \hyll much more natural than in
ordinary temporal logics.
One important example is of \emph{oscillation} between states in systems with
kinetic feedback. In a temporal specification language such as
BIOCHAM~\cite{chabrier05cmsb}, only aperiodic oscillations are representable,
while in \hyll an oscillation between "A" and "B" with delay "d" is represented
by the rule "!! (A -o rate d B) with (B -o rate d A)" (or "!! (A -o dia B) with
(B -o dia A)" if the oscillation is aperiodic).
If \hyllt were extended with constrained implication and conjunction in the
style of CILL~\cite{saranli07icra} or $\eta$~\cite{deyoung08csf}, then we can
define localized versions of "box" and "dia", such as ``"A" is true
everywhere/somewhere in an interval''.

For examples of applications of HyLL with temporal constraints, the interested
reader can see~\cite{deMaria-Despeyroux-Felty:14rep}, which gives an encoding of
a simple biological system and its temporal properties in \chyll[T'], where
"\mathcal{T'} = \langle \mathbb{N}, +, 0\rangle" represents discrete instants of
time.
We will, instead, use a version of HyLL dedicated to continuous time Markov
Chains with exponential distribution, as used in S$\pi$.
We introduce this type of constraints below.

\subsection{Probabilistic Constraints}
\label{sec:hyllp-pi}

Transitions in practice rarely have precise delays.
Phenomenological and experimental evidence is used to construct a probabilistic
model of the transition system where the delays are specified as probability
distributions of continuous (or discrete) variables.
A number of variations of monoids representing probabilistic and stochastic
constraints are presented in~\cite{Chaudhuri-Despeyroux:13tr}, both for the
general case and for the special case of Markov processes.

One of the standard models of stochastic transition systems is continuous time
Markov chains (CTMCs) where the delays of transitions between states are
distributed according to the Markov assumption of memorylessness (Markov
processes) with the further condition that their state-spaces are
countable~\cite{Rogers-Williams-vol1-book}.
In the synchronous stochastic $\pi$-calculus (S$\pi$), the probability of a
reaction with \emph{rate} $r$ is given by continuous time Markov chains with
exponential distribution of parameter $r$ (See \cite{phillips06tcsb}).
To describe such processes, we shall take $\Reals^+$ to represent the
\emph{rates} of their exponential distribution.
To encode the S$\pi$ calculus in a suitable instanciation of HyLL, we only need
a \emph{symbolic} operation on the rates.
This abstract treatment can be made fully precise, but this would require a
detour into measure theory that is beyond the scope of this paper;
see~\cite{Chaudhuri-Despeyroux:13tr} for the details.

\begin{definition} \label{defn:rates}%
  The \emph{rates domain} $\mathcal{R}$ is the monoid $\mathcal{R} = \langle
  {\Reals^+}^*, ".", "nil"\rangle$ of lists of positive reals, where "." is
  concatenation of lists, and "nil" is the empty list.
\end{definition}

\noindent%
Worlds $r \in \mathcal{R}$ represent the (rates of the) \emph{sequence} of
actions that have led to the current world from a given fixed initial world.
We might equivalently have chosen $\mathcal{T} = \langle \Reals^+, +,0\rangle$ to
represent average time delays, which would be the sum of the reciprocals of the
rates in the list of rates.
We choose to use lists of rates because they are more informative than average
time delays.
Note that since our rate functions are assumed to be memoryless, the order of
the list of rates is immaterial, so we can easily relax it to a multi-set of
rates; this change could not substantially alter the development of this paper.

\section{Focusing}
\label{sec:focusing}

As \hyll is intended to represent transition systems adequately, it is crucial
that \hyll derivations in the image of an encoding have corresponding
transitions. However, transition systems are generally specified as rewrite
algebras over an underlying congruence relation. These congruences have to be
encoded propositionally in \hyll, so a \hyll derivation will generally require
several inference rules to implement a single transition; moreover, several
trivially different reorderings of these ``micro'' inferences would correspond
to the same transition. It is therefore futile to attempt to define an
operational semantics directly on \hyll inferences.

We restrict the syntax to focused derivations~\cite{andreoli92jlc}, which
ignores many irrelevant rule permutations in a sequent proof and divides the
proof into clear \emph{phases} that define the grain of atomicity. The logical
connectives are divided into two classes, \emph{negative} and \emph{positive},
and rule permutations for connectives of like polarity are confined to
\emph{phases}. A \emph{focused derivation} is one in which the positive and
negative rules are applied in alternate maximal phases in the following way: in
the \emph{active} phase, all negative rules are applied (in irrelevant order)
until no further negative rule can apply; the phase then switches and one
positive proposition is selected for \emph{focus}; this focused proposition is
decomposed under focus (\ie, the focus persists to its sub-formulas) until it
becomes negative, and the phase switches again.

As noted before, the logical rules of the hybrid connectives "at" and "dn" are
invertible, so they can be considered to have both polarities. It would be valid
to decide a polarity for each occurrence of each hybrid connective
independently; however, as they are mainly intended for book-keeping during
logical reasoning, we define the polarity of these connectives in the following
\emph{parasitic} form: if its immediate subformula is positive (resp. negative)
connective, then it is itself positive (resp. negative). These connectives
therefore become invisible to focusing. This choice of polarity can be seen as a
particular instance of a general scheme that divides the "dn" and "at"
connectives into two polarized forms each. To complete the picture, we also
assign a polarity for the atomic propositions; this restricts the shape of
focusing phases further~\cite{chaudhuri08jar}. The full syntax of positive ($P,
Q, \ldots$) and negative ($M, N, \ldots$) propositions is as follows:

\medskip
\bgroup %\small
\hspace{-1.5em}
\begin{tabular}{l@{\ }r@{\ \ }l}
  "P, Q, ..." & "::=" & "p~\vec t OR P tens Q OR one OR P plus Q OR zero OR {! N} OR \fex \alpha P OR now u. P OR (P at w) OR pos N" \\
  "N, M, ..." & "::=" & "n~\vec t OR N with N OR top OR P -o N OR \fall \alpha N OR now u. N OR (N at w) OR neg P"
\end{tabular}
\egroup

\begin{figure}[p]
\centering
\framebox{
\begin{minipage}{.95\textwidth}
  \setlength{\parindent}{0pt}
  \bgroup \small
%%  \paragraph{Focused logical rules}
  {\bf \large Focused logical rules}
  \begin{gather*}
    \I[li]{"\G ; foc{n\ \vec t @ w} ==> pos {n\ \vec t @ w}"}
    \quad
    \I["neg L"]{"\G ; \D ; foc{neg P @ u} ==> Q @ w"}{"\G ; \D ; P @ u ==> . ; Q @ w"}
    \quad
    \I["pos R"]{"\G ; \D ==> foc{pos N @ w}"}{"\G ; \D ; . ==> N @ w ; ."}
    \\[1ex]
    \I["with L_i"]{"\G ; \D ; foc{N_1 with N_2 @ u} ==> Q @ w"}{
      "\G ; \D ; foc{N_i @ u} ==> Q @ w"
    }
    \quad
    \I["{-o}L"]{"\G ; \D, \X ; foc{P -o N @ u} ==> Q @ w"}{
      "\G ; \D ==> foc{P @ u}" & "\G ; \X ; foc{N @ u} ==> Q @ w"
    }
    \\[1ex]
    \I["\forall L"]{"\G ; \D ; foc {\fall \alpha N @ u} ==> Q @ w"}{
      "\G ; \D ; foc {[\tau / \alpha] N @ u} ==> Q @ w"
    }
    \quad
    \I["{dn}LF"]{"\G ; \D ; foc{now u. N @ v} ==> Q @ w"}{"\G ; \D ; foc{[v / u] N @ v} ==> Q @ w"}
    \\[1ex]
    \I["{at}LF"]{"\G ; \D ; foc{(N at u) @ v} ==> Q @ w"}{"\G ; \D ; foc{N @ u} ==> Q @ w"}
    \quad
    \I[ri]{"\G ; neg {p\ \vec t} @ w ==> foc {p\ \vec t @ w}"}
    \\[1ex]
    \I["tens R"]{"\G ; \D, \X ==> foc{P tens Q @ w}"}{
      "\G ; \D ==> foc{P @ w}" & "\G ; \X ==> foc{Q @ w}"
    }
    \quad
    \I["plus R_i"]{"\G ; \D ==> foc{P_1 plus P_2 @ w}"}{"\G ; \D ==> foc{P_i @ w}"}
    \\[1ex]
    \I["\exists R"]{"\G ; \D ==> foc{\fex \alpha P @ w}"}{"\G ; \D ==> foc{[\tau/\alpha] P @ w}"}
    \quad
    \I["!R"]{"\G ; . ==> foc{{!N} @ w}"}{"\G ; . ; . ==> N @ w ; ."}
    \quad
    \I["{dn}RF"]{"\G ; \D ==> foc{now u. P @ w}"}{"\G ; \D ==> foc{[w/u] P @ w}"}
    \\[1ex]
    \I["at RF"]{"\G ; \D ==> foc{(P at u) @ w}"}{"\G ; \D ==> foc{P @ u}"}
    \quad
    \I["one R"]{"\G ; . ==> foc {one @ w}"}
  \end{gather*}

%%  \paragraph{Active logical rules}
  {\bf \large Active logical rules}

  ("\RR" of the form ". ; Q @ w" or "N @ w ; .", and "\LL" of the form
  "\G ; \D ; \W")
  \begin{gather*}
    \I["tens L"]{"\LL, P tens Q @ u ==> \RR"}{"\LL, P @ u, Q @ u ==> \RR"}
    \quad
    \I["one L"]{"\LL, one @ u ==> \RR"}{"\LL ==> \RR"}
    \quad
    \I["plus L"]{"\LL, P plus Q @ u ==> \RR"}{
      "\LL, P @ u ==> \RR" & "\LL, Q @ u ==> \RR"
    }
    \\[1ex]
    \I["{dn}LA"]{"\LL, now u. P @ v ==> \RR"}{"\LL, [v/u] P @ v ==> \RR"}
    \quad
    \I["{at}LA"]{"\LL, (P at u) @ v ==> \RR"}{"\LL, P @ u ==> \RR"}
    \quad
    \I["\exists L^\alpha"]{"\LL, \fex \alpha P @ u ==> \RR"}{"\LL, P @ u ==> \RR"}
    \\[1ex]
    \I["{!L}"]{"\G ; \D ; \W, {!N @ u} ==> \RR"}{"\G, N @ u ; \D ; \W ==> \RR"}
    \quad
    \I["pos L"]{"\G ; \D ; \W, {pos N @ w} ==> \RR"}{"\G ; \D, N @ w ; \W ==> \RR"}
    \quad
    \I[lp]{"\G ; \D ; \W, {p\ \vec t} @ w ==> \RR"}{"\G ; \D, neg p\ \vec t ; \W ==> \RR"}
    \\[1ex]
    \I["with R"]{"\LL ==> M with N @ w ; ."}{"\LL ==> M @ w ; ." & "\LL ==> N @ w ; ."}
    \quad
    \I["top R"]{"\LL ==> top @ w ; ."}
    \quad
    \I["{-o}R"]{"\LL ==> P -o N @ w ; ."}{"\LL, P @ w ==> N @ w ; ."}
    \\[1ex]
    \I["{dn}RA"]{"\LL ==> now u. N @ w ; ."}{"\LL ==> [w/u] N @ w ; ."}
    \quad
    \I["{at}RA"]{"\LL ==> (N at u) @ w"}{"\LL ==> N @ u"}
    \quad
    \I["\forall R^\alpha"]{"\LL ==> \fall \alpha N @ u ; ."}{"\LL ==> N @ u ; ."}
    \\[1ex]
    \I["neg R"]{"\LL ==> {neg P @ w} ; ."}{"\LL ==> . ; P @ w"}
    \quad
    \I[rp]{"\LL ==> n\ \vec t @ w ; ."}{"\LL ==> . ; pos n\ \vec t @ w"}
    \quad
    \I["zero L"]{"\LL, zero @ u ==> \RR"}
  \end{gather*}

%%  \paragraph{Focusing decisions}
  {\bf \large Focusing decisions}
%
%%  ("\LL" of the form "\G ; \D")
  %
  \begin{gather*}
    \I[lf]{"\G ; \D, N @ u ; . ==> . ; Q @ w"}{"\G ; \D ; foc {N @ u} ==> Q @ w" & N \text{ not } "neg p\ \vec t"}
    \quad
    \I[cplf]{"\G, N @ u ; \D ; . ==> . ; Q @ w"}{"\G, N @ u ; \D ; foc {N @ u} ==> Q @ w"}
    \\[1ex]
    \I[rf]{"\G ; \D ; . ==> . ; P @ w"}{"\G ; \D ==> foc {P @ w}" & P \text{ not } "pos n\ \vec t"}
  \end{gather*}
  \egroup
  \end{minipage}}
\caption{Focusing rules for \hyll.}
\label{fig:foc-rules}
\end{figure}

\medskip
\noindent
The two syntactic classes refer to each other via the new \emph{shift}
connectives "neg" and "pos". Sequents in the focusing calculus are of the
following forms.
\begin{center} \small
  \begin{tabular}{r@{\ }l@{\qquad}r@{\ }l}
    \( \left.
    \begin{array}[c]{l}
      "\G ; \D ; \W ==> . ; P @ w" \\
      "\G ; \D ; \W ==> N @ w ; ."
    \end{array}
    \right\} \) & active &
    \( \left.
    \begin{array}[c]{l}
      "\G ; \D ; foc{N @ u} ==> P @ w" \\
      "\G ; \D ==> foc{P @ w}" \\
    \end{array}
    \right\} \) & focused \\
  \end{tabular}
\end{center}
In each case, "\G" and "\D" contain only negative propositions (\ie, of the form
"N @ u") and "\W" only positive propositions (\ie, of the form "P @ u").
The full collection of inference rules are in \figref{foc-rules}.
The sequent form "\G ; \D ; . ==> . ; P @ w" is called a \emph{neutral sequent}; from
such a sequent, a left or right focused sequent is produced with the rules lf,
cplf or rf. Focused logical rules are applied (non-deterministically) and focus
persists unto the subformulas of the focused proposition as long as they are of
the same polarity; when the polarity switches, the result is an active sequent,
where the propositions in the innermost zones are decomposed in an irrelevant
order until once again a neutral sequent results.

Soundness of the focusing calculus with respect to the ordinary sequent calculus
is immediate by simple structural induction. In each case, if we forget the
additional structure in the focused derivations, then we obtain simply an
unfocused proof. We omit the obvious theorem. Completeness, on the other hand,
is a hard result. We omit the proof because focusing is by now well known for
linear logic, with a number of distinct proofs via focused cut-elimination (see
\eg the detailed proof in~\cite{chaudhuri08jar}). The hybrid connectives pose no
problems because they allow all cut-permutations.

\bgroup %\small
\begin{theorem}[focusing completeness]
  Let "\G^-" and "C^- @ w" be negative polarizations of "\G" and "C @ w" (that
  is, adding "neg" and "pos" to make "C" and each proposition in "\G" negative)
  and "\D^+" be a positive polarization of "\D". If "\G ; \D ==> C @ w", then
  ". ; . ; {!  \G^-}, \D^+ ==> C^- @ w ; .".
\end{theorem}
\egroup

\section{Encoding the Synchronous Stochastic $\pi$-calculus}
\label{sec:spi}

In this section, we shall illustrate the use of \hyllr (Definiition
\ref{defn:rates}) as a logical framework for constrained transition systems by
encoding the syntax and the operational semantics of the synchronous stochastic
$\pi$-calculus (\spi), which extends the ordinary $\pi$-calculus by assigning to
every channel and internal action an \emph{inherent} rate of synchronization.
In \spi, each rate characterises an exponential distribution
such that the probability of a reaction with rate $r$
%% (more generally a transition with rate $r$)
occuring within time $t$ is given by
%% WAS:  $\mathbf{F}(t) =
$1 - e^{-rt}$ \cite{phillips06tcsb},
where the \emph{rate} $r$ is a parameter.
% The expected value of such distributions are $r^{-1}$. %% True but useless here.
%
%% From
%% G. Norman, C. Palamidessi, D. Parker, P. Wu.
%% Model checking probabilistic and stochastic extensions of the π-calculus.
%% IEEE Transactions of Software Engineering 35(2): 209-223, 2009.
%
% In other words, for rate $r$,
% the probability that the transition is enabled within $t$ time-units is given by
% $1 - e^{-rt}$.
%
We shall encode \spi in \hyllr:
a \spi reaction with rate $r$ will be encoded by a transition
of a probability described by a random variable  %% having an
with exponential distribution of parameter $r$;
Worlds $r$ in $\mathcal{R}$ will represent the list of the \emph{rates} of the
% exponential distributions of the variables associated to the
transitions performed so far.

\hyllr can therefore be seen as a formal language for expressing \spi executions (traces).
For the rest of this section we shall use "r, s, t, \ldots" instead of "u, v, w,
\ldots" to highlight the fact that the worlds represent (lists of) rates
(overloading single elements and the list of single elements).
We do not directly use rates because the syntax and transitions of \spi are
given generically for a $\pi$-calculus with labelled actions, and it is only
the interpretation of the labels that involves probabilities.

We first summarize the
syntax of \spi, which is a minor variant of a number of similar presentations
such as~\cite{phillips06tcsb}. For hygienic reasons we divide entities into the
syntactic categories of \emph{processes} ($\pr P, \pr Q, \ldots$) and
\emph{sums} ($\pr M, \pr N, \ldots$), defined as follows. We also include
environments of recursive definitions ("pr E") for constants.

\smallskip
\bgroup %\small
\begin{tabular}{l@{\quad}l@{\ }r@{\ }l}
\emph{(Processes)} & "pr P, pr Q, ..." & "::=" & "pr{\nu_r\ P} OR pr{P par Q} OR pr 0 OR pr{X_n\,x_1 \cdots x_n} OR pr M" \\
\emph{(Sums)} & "pr M, pr N, ..." & "::=" & "pr {act{oup{x}(y)} P} OR pr {act{inp{x}} P} OR pr{act{\tau_r} P} OR pr{M + N}" \\
\emph{(Environments)} & "pr E" & "::=" & "pr {E, X_n \triangleq P} OR pr ."
\end{tabular}
\egroup

\smallskip

"pr{P par Q}" is the parallel composition of "pr P" and "pr Q", with unit "pr
0". The restriction "pr{\nu_r\ P}" abstracts over a free channel "x" in the
process "pr{P\,x}". We write the process using higher-order abstract
syntax~\cite{pfenning88pldi}, \ie, "pr{P}" in "pr{\nu_r\ P}" is (syntactically)
a function from channels to processes. This style lets us avoid cumbersome
binding rules in the interactions because we reuse the well-understood binding
structure of the $\lambda$-calculus. A similar approach was taken in the
earliest encoding of the (ordinary) $\pi$-calculus in (unfocused) linear
logic~\cite{miller92welp}, and is also present in the encoding in
CLF~\cite{cervesato03tr}.

A sum is a non-empty choice ("+") over terms with \emph{action prefixes}: the
output action "oup{x}(y)" sends "y" along channel "x", the input action "inp{x}"
reads a value from "x" (which is applied to its continuation process), and the
internal action "\tau_r" has no observable I/O behaviour. Replication of
processes happens via guarded recursive definitions~\cite{milner99book};
in~\cite{Regev01psb} it is argued that they are more practical for programming
than the replication operator "!". In a definition "pr{X_n \triangleq P}", "pr
{X_n}" denotes a (higher-order) defined constant of arity "n"; given channels
"x_1, ..., x_n", the process "pr {X_n\,x_1 \cdots x_n}" is synonymous with
"pr{P\,x_1 \cdots x_n}". The constant "pr{X_n}" may occur on the right hand side
of any definition in "pr E", including in its body "pr P", as long as it is
prefixed by an action; this prevents infinite recursion without progress.

Interactions are of the form "pr E |- pr P -> [r] pr Q" denoting a transition
from the process "pr P" to the process "pr Q", in a global environment "pr E",
by performing an action at rate "r". Each channel "x" is associated with an
inherent rate specific to the channel, and internal actions "\tau_r" have rate
"r". The restriction "pr{\nu_r\ P}" defines the rate of the abstracted channel
as "r".

\sdef{crate}{\mathop{\mathrm{rate}}}

\begin{figure*}[tp]
\centering
%\hspace{-1.8em}
\framebox{ \begin{minipage}{.97\linewidth}
  \smaller[2]
  \emph{Interactions} \vspace{-1em}

  \begin{gather*}
    \I[\set{SYN}]{"pr{act{oup{x}(y)} P + M par act{inp{x}} Q + M'} ->[crate(x)] pr{P par Q\,y}"}{}
    \SP
    \I[\set{INT}]{"pr{act{\tau_r} P} ->[r] pr P"}{}
    \\
    \I[\set{PAR}]{"pr{P par Q} ->[r] pr{P' par Q}"}{"pr P ->[r] pr{P'}"}
    \SP
    \I[\set{RES}]{"pr{\nu_s\ P} ->[r] pr{\nu_s\ Q}"}{
      \forall x_s. \Bigl("pr{P\,x} ->[r] pr{Q\,x}"\Bigr)
    }
    \SP
    \I[\set{CONG}]{"pr{P'} ->[r] pr{Q'}"}{"pr{P} ->[r] pr{Q}" & "pr{P} == pr{P'}" & "pr{Q} == pr{Q'}"}
  \end{gather*}

  \vspace{-1.5em} \dotfill

  \emph{Congruence} \vspace{-1em}
  \begin{gather*}
    \I{"pr{P par 0} == pr P"}{}
    \LSP
    \I{"pr{P par Q} == pr{Q par P}"}{}
    \LSP
    \I{"pr{P par (Q par R)} == pr{(P par Q) par R}"}{}
    \LSP
    \I{"pr{\nu_r\ 0 == 0}"}{}
    \LSP
    \I{"pr{E} |- pr{X_n\,x_1 \cdots x_n} == pr{P\,x_1 \cdots x_n}"}
      {"pr{X_n \triangleq P} \in pr{E} "}
    \\
    \I{"pr{\nu_r (lam x. \nu_s (lam y. P))} == pr{\nu_s (lam y. \nu_r (lam x. P))}"}{}
    \LSP
    \I{"pr{\nu_r\ P} == pr{\nu_r\ Q}"}{
      \forall x_r.\left("pr{P\,x} == pr{Q\,x}"\right)
    }
    \LSP
    \I{"pr{\nu_r (lam x. P par Q(x))} == pr{P par \nu_r\ Q}"}
    \\
    \I{"pr{P par Q} == pr{P' par Q}"}{"pr P == pr{P'}"}
    \LSP
    \I{"pr{act{oup{x}(m)} P} == pr{act{oup{x}(m)} P'}"}{"pr P == pr {P'}"}
    \LSP
    \I{"pr{act{inp{x}} P} == pr{act{inp{x}} Q}"}{
      \forall n.\ \left("pr{P\,n} == pr{Q\,n}"\right)
    }
    \LSP
    \I{"pr{act{\tau_r} P} == pr{act{\tau_r} P'}"}{
      "pr P == pr {P'}"
    }
    \LSP
    \I{"pr{M + N} == pr{N + M}"}
    \\
    \I{"pr{M + (N + K)} == pr{(M + N) + K}"}
    \LSP
    \I{"pr{M + N} == pr{M' + N}"}
      {"pr M == pr {M'}"}
    \LSP
    \I{"pr{M + N} == pr M"}{"pr M == pr N"}
  \end{gather*}
\end{minipage}}
\caption{Interactions and congruence in \spi. The environment $E$ is elided in most rules.}
\label{fig:spi}
\end{figure*}

The full set of interactions and congruences are in fig.~\ref{fig:spi}. We
generally omit the global environment "pr E" in the rules as it never changes.
It is possible to use the congruences to compute a normal form for processes
that are a parallel composition of sums and each reaction selects two suitable
sums to synchronise on a channel until there are no further reactions possible;
this refinement of the operational semantics is used in "spi" simulators such as
SPiM~\cite{phillips04bc}.

\bgroup %\small
\begin{definition}[syntax encoding] \mbox{} \label{defn:sencoding}
  \begin{ecom}[1.]
  \item The encoding of the process "pr P" as a positive proposition, written
    "proc P", is as follows ("dt" is a positive atom and \crt a negative atom).
    \begin{align*}
      "proc{P par Q}" &= "proc P tens proc Q"
      & \SP
      "proc {\nu_r\ P}" &= "ex x. ! (rt x at r) tens proc {P\,x}"
%%      "proc {\nu_r\ P}" &= "ex x. ! (rt x at [r]) tens proc {P\,x}"
      \\
      "proc 0" &= "one"
      &
      "proc{X_n\,x_1 \cdots x_n}" &= "X_n\,x_1 \cdots x_n"
      \\
      "proc M" &= "pos {(dt -o sum M)}"
    \end{align*}
  \item The encoding of the sum "pr M" as a negative proposition, written "sum
    M", is as follows (\cn{out}, \cn{in} and \cn{tau} are positive atoms).
    \begin{align*}
      "sum {M + N}" &= "sum M with sum N"
      &\hspace{-2em}
      "sum {act{oup x(m)} P}" &= "neg (\cout x~m tens proc P)"
      \\
      "sum {act{inp x} P}" &= "all n. neg (\cin x~n tens proc {P\,n})"
      &
      "sum {act{\tau_r} P}" &= "neg (\ctau ~r tens proc P)"
%%      "sum {act{\tau_r} P}" &= "neg (\ctau ~[r] tens proc P)"
    \end{align*}
  \item The encoding of the definitions "pr E" as a context, written "env E", is
    as follows.
    \begin{align*}
      "env{E, X_n \triangleq P} " \ & = \
         "env{E}, !! \fall{x_1, ..., x_n} X_n\,x_1 \cdots x_n o-o proc{P\,x_1 \cdots x_n} " \\
      "env{ . } " \ & = \  " . "
    \end{align*}
    where "P o-o Q" is defined as "(P -o neg Q) with (Q -o neg P)".
  \end{ecom}
\end{definition}
\egroup

The encoding of processes is positive, so they will be decomposed in the active
phase when they occur on the left of the sequent arrow, leaving a collection of
sums. The encoding of restrictions will introduce a fresh unrestricted
assumption about the rate of the restricted channel. Each sum encoded as a
processes undergoes a polarity switch because "-o" is negative; the antecedent
of this implication is a \emph{guard} "dt". This pattern of guarded switching of
polarities prevents unsound congruences such as "pr{act{oup x(m)} act{oup y(n)}
  P} == pr{act{oup y(n)} act{oup x(m)} P}" that do not hold for the synchronous
$\pi$ calculus. To see this, note that "proc{pr{act{oup x(m)} act{oup y(n)} P}}"
has the form "X -o (A tens (X -o B tens C))" (eliding the polarity shifts) which
is not provably equivalent to "X -o (B tens (X -o A tens C))" in both linear
logic and \hyll. Thus, even though we use a commutative connective "tens" in
"sum{act{oup x(m)} P}", output actions are still sequential and synchronous.

The guard "dt" also \emph{locks} the sums in the context: the \spi interaction
rules \set{INT} and \set{SYN} discard the non-interacting terms of the sum, so
the environment will contain the requisite number of "dt"s only when an
interaction is in progress.
The action prefixes themselves are also synchronous, which causes another
polarity switch. Each action releases a token of its respective kind---\cn{out},
\cn{in} or \cn{tau}---into the context. These tokens must be consumed by the
interaction before the next interaction occurs. For each action, the (encoding
of the) continuation process is also released into the context.

The proof of the following congruence lemma is omitted. Because the encoding is
(essentially) a "tens / with" structure, there are no distributive laws in
linear logic that would break the process/sum structure.

\begin{theorem}[congruence]
  \label{thm:congr} %\hfill\break
  "pr E |- pr P == pr Q" iff both "env E @ rid ; . ; proc P @ rid ==> . ; proc Q @ rid"
  and "env E @ rid ; . ; proc Q @ rid ==> . ; proc P @ rid".
\end{theorem}

Now we encode the interactions. Because processes were lifted into propositions,
we can be parsimonious with our encoding of interactions by limiting ourselves
to the atomic interactions \set{syn} and \set{int} (below); the \set{par},
\set{res} and \set{cong} interactions will be ambiently implemented by the
logic. Because there are no concurrent interactions---only one interaction can
trigger at a time in a trace---the interaction rules must obey a locking
discipline. We represent this lock as the proposition \cact that is consumed at
the start of an interaction and produced again at the end. This lock also
carries the net rate of the prefix of the trace so far: that is, an interaction
"pr P ->[r] pr Q" will update the lock from "\cact @ s" to "\cact @ {s .
  r}". The encoding of individual atomic interactions must also remove the
\cn{in}, \cn{out} and \cn{tau} tokens introduced in context by the interacting
processes.

%% \bgroup %\small
\begin{definition}[interaction] \mbox{} \newline
\label{defn:iencoding}
  Let "\cinter \triangleq !! (\cact -o neg \cint with neg \csyn)" where \cact is
  a positive atom and \cint and \csyn are as follows:
  \bgroup\small
  \begin{align*}
    \cint &\triangleq "(dt at rid) tens pos all r. \Bigl((\ctau r at rid) -o rate r neg \cact\Bigr)" \\
    \csyn &\triangleq "(dt tens dt at rid) tens pos all x, r, m. \Bigl((\cout x~m tens \cin x~m at rid) -o pos (rt x at r) -o rate r neg \cact \Bigr)".
   \end{align*}
   \egroup
\end{definition}
%% \egroup

\noindent
The number of interactions that are allowed depends on the number of instances of
\cinter in the linear context: each focus on \cinter implements a single
interaction. If we are interested in all finite traces, we will add \cinter to
the unrestricted context so it may be reused as many times as needed.

\subsection{Representational Adequacy.}
\label{sec:spi.adq}

Adequacy consists of two components: completeness and soundness. Completeness is
the property that every \spi execution is obtainable as a \hyll derivation using
this encoding, and is the comparatively simpler direction (see
\thmref{completeness}). Soundness is the reverse property, and is false for
unfocused \hyll as such. However, it \emph{does} hold for focused proofs (see
\thmref{adeq}). In both cases, we reason about the following canonical sequents
of \hyll.

\begin{definition} % \small
  The \emph{canonical context of} "P", written "can P", is given by:
  \begin{align*} %% \small
    "can {X_n\,x_1 \cdots x_n}" &= "neg {X_n\,x_1 \cdots x_n}" &
    "can {P par Q}" &= "can P", "can Q" &
    "can 0" &= "." &
    "can {\nu_r\ P}" &= "can{P\,a}" \\
    "can M" &= "dt -o sum M" &
  \end{align*}
  For "can{\nu_r\ P}", the right hand side uses a \emph{fresh} channel "a" that
  is not free in the rest of the sequent it occurs in.
\end{definition}

\noindent
As an illustration, take "pr P \triangleq pr{act{oup x(a)} Q par act{inp x} R}". We have:
\begin{gather*} % \small
  "can P" =
  \begin{array}[t]{l}
    "dt -o neg (\cout x~a tens proc Q)",
    "dt -o all y. neg (\cin x~y tens  proc {R\,y})"
  \end{array}
\end{gather*}
Obviously, the canonical context is what would be emitted to the linear zone at
the end of the active phase if "proc P" were to be present in the left active
zone.

\begin{definition} % \small
  A neutral sequent is \emph{canonical} iff it has the shape
  \begin{gather*} %% \small
    "env{E}, \cn{rates}, \cinter @ rid ;
     neg \cact @ s, can {P_1 par \cdots par P_k} @ rid ;
     . ==> . ; (proc Q at rid) tens \cact @ t"
  \end{gather*}
  where \cn{rates} contains elements of the form "rt x @ r" defining the
  rate of the channel "x" as "r", and all free channels in "env{E}, can{P_1 par \cdots
    par P_k par Q}" have a single such entry in \cn{rates}.
\end{definition}

\begin{figure*}[tp]
  \centering
  \small
  \begin{multline*}
    \text{Suppose "\LL = \crt x @ r, \cinter @ rid" and "\RR = (proc S
      at rid) tens \cact @ t". (All judgements "{} @ rid" omitted.)}
    \\[1ex]
    \hspace{-1em}
    \infer={"\LL ; neg \cact @ s, can{act{oup x(a)} Q par act{inp x} R} ; . ==> . ; \RR"}{
    \I[1]{"\LL ;
              neg \cact @ s,
              dt -o neg (\cout x\ a tens proc Q),
              dt -o all y. neg (\cin x~y tens proc {R\,y})
              ; . ==> . ; \RR"}
         {\I[2]{"\LL ;
                     neg \cact @ s,
                     dt -o neg (\cout x\ a tens proc Q),
                     dt -o all y. neg (\cin x~y tens proc {R\,y})
                     ; foc{\cinter} ==> \RR"
               }
               {\I[3]{\LL\ ;\
                      \begin{array}[t]{l}
                        "dt -o neg (\cout x\ a tens proc Q)",
                        "dt -o all y. (\cin x~y tens proc {R\,y})", \\
                        "neg dt, neg dt, all x, r, m. ((\cout x~m tens \cin x~m at rid) -o pos(\crt x at r) -o rate {r} \cact) @ s"
                      \end{array}
                      \hspace{-1em}
                      \begin{array}[t]{l}
                        {} \\ "{} ; . ==> . ; \RR"
                      \end{array}
                     }
                     {\I[4]{\LL\ ; \
                            \begin{array}[t]{l}
                              "neg \cout x\ a", "can Q",
                              "dt -o all y. neg (\cin x~y tens proc {R\,y})", \\
                              "neg dt, all x, r, m. ((\cout x~m tens \cin x~m at rid) -o pos(\crt x at r) -o rate {r} \cact) @ s"
                            \end{array}
                            \hspace{-1em}
                            \begin{array}[t]{l}
                              {} \\ "{} ; . ==> . ; \RR"
                            \end{array}
                           }
                           {\I[5]{\LL\ ;\
                                  \begin{array}[t]{l}
                                    "can Q", "neg \cout x~a", "neg \cin x~a", "can{R\,a}", \\
                                    "all x, r, m. ((\cout x~m tens \cin x~m at rid) -o pos(\crt x at r) -o rate {r} \cact) @ s"
                                  \end{array}
                                  \hspace{-1em}
                                  \begin{array}[t]{l}
                                    {} \\ "{} ; . ==> . ; \RR"
                                  \end{array}
                                 }
                                 {"\LL ; can{Q}, can{R\,a}, neg \cact @ {s . r} ; . ==> . ; \RR"}
                           }
                     }
               }
         }
       }\\[-2em]
  \end{multline*}
  \begin{tabular}{l@{:\ }l@{\SP}l@{:\ }l@{\SP}l@{:\ }l}
    \multicolumn{2}{l}{Steps}\\\hline
    1 & focus on "\cinter \in \LL" &
    3 & "dt" for output + full phases &
    5 & cleanup \\
    2 & select \csyn from \cinter, active rules &
    4 & "dt" for input + full phases
  \end{tabular}
  \caption{Example interaction in the \spi-encoding.}
  \label{fig:example-syn}
\end{figure*}

\figref[Figure]{example-syn} contains an example of a derivation for a canonical
sequent involving "pr P". Focusing on any (encoding of a) sum in "can P @ rid"
will fail because there is no "dt" in the context, so only \cinter can be given
focus; this will consume the \cact and release two copies of "(dt at rid)" and
the continuation into the context. Focusing on the latter will fail now (because
"\cout x~m" and "\cin x~m" (for some "m") are not yet available), so the only
applicable foci are the two sums that can now be ``unlocked'' using the
"dt"s. The output and input can be unlocked in an irrelevant order, producing
two tokens "\cin x~a" and "\cout x~a". Note in particular that the witness "a"
was chosen for the universal quantifier in the encoding of "pr{act{inp{x}}Q}"
because the subsequent consumption of these two tokens requires the messages to
be identical. (Any other choice will not lead to a successful proof.)  After
both tokens are consumed, we get the final form "\cact @ {s . r}", where "r" is
the inherent rate of "x" (found from the \cn{rates} component of the
unrestricted zone). This sequent is canonical and contains "can{Q par R\,a}".

Our encoding therefore represents every \spi action in terms of ``micro''
actions in the following rigid order: one micro action to determine what kind of
action (internal or synchronization), one micro action per sum to select the
term(s) that will interact, and finally one micro action to establish the
contract of the action. Thus we see that focusing is crucial to maintain the
semantic interpretation of (neutral) sequents. In an unfocused calculus, several
of these steps could have partial overlaps, making such a semantic
interpretation inordinately complicated. We do not know of any encoding of the
$\pi$ calculus that can provide such interpretations in unfocused sequents
without changing the underlying logic. In CLF~\cite{cervesato03tr} the logic is
extended with explicit monadic staging, and this enables a form of
adequacy~\cite{cervesato03tr}; however, the encoding is considerably more
complex because processes and sums cannot be fully lifted and must instead be
specified in terms of a lifting computation. Adequacy is then obtained via a
permutative equivalence over the lifting operation. Other encodings of $\pi$
calculi in linear logic, such as~\cite{garg05concur} and~\cite{baelde05stage},
concentrate on the easier asynchronous fragment and lack adequacy proofs anyhow.

\begin{theorem}[completeness] \label{thm:completeness} % \small
  If "pr E |- pr P ->[r] pr Q", then the following canonical sequent is derivable.
  \begin{gather*}
    "env{E}, \cn{rates}, \cinter @ rid ; neg \cact @ s, can P @ rid ; . ==> . ; (proc Q at rid) tens \cact @ {s . r}".
  \end{gather*}
\end{theorem}

\begin{proof}
  By structural induction of the derivation of "pr E |- pr P ->[r] pr Q". Every
  interaction rule of \spi is implementable as an admissible inference rule for
  canonical sequents. For \set{cong}, we appeal to \thmref{congr}.
\end{proof}

Completeness is a testament to the expressivity of the logic -- all executions
of \spi are also expressible in \hyll. However, we also require the opposite
(soundness) direction: that every canonical sequent encodes a possible \spi
trace. The proof hinges on the following canonicity lemma.

\begin{lemma}[canonical derivations] \label{lem:canonical} % \small
  In a derivation for a canonical sequent, the derived inference rules for
  \cinter are of one of the two following forms (conclusions and premises
  canonical).
  \begin{gather*}
    \I{"env{E}, \cn{rates}, \cinter @ rid ; neg \cact @ s, can P @ rid ; . ==> . ; (proc P at rid) tens \cact @ s"}{}
    \\[1ex]
    \I{"env{E}, \cn{rates}, \cinter @ rid ; neg \cact @ s, can P @ rid ; . ==> . ; (proc R at rid) tens \cact @ t"}
      {"env{E}, \cn{rates}, \cinter @ rid ; neg \cact @ {s . r}, can Q @ rid ; . ==> . ; (proc R at rid) tens \cact @ t"}
  \end{gather*}
%% WAS:  where: either "pr E |- pr P ->[r] pr Q", or "pr E |- pr P == pr Q" with "r = rid".
where: either "pr E |- pr P == pr Q" with "r = rid" or "pr E |- pr P ->[r] pr Q".
\end{lemma}

\begin{proof}
  This is a formal statement of the phenomenon observed earlier in the example
  (\figref{example-syn}): "proc R tens \cact" cannot be focused on the right
  unless "pr P == pr R", in which case the derivation ends with no more foci on
  \cinter. If not, the only elements available for focus are \cinter and one of
  the congruence rules "env{E}" in the unrestricted context.
%% Former and latter cases WAS exchanged:
  In the former case,
  the definition of a top level "X_n" in "can P" is (un)folded (without advancing
  the world).
  In the latter case,
  the derived rule consumes the "neg \cact @ s", and by the time \cact is
  produced again, its world has advanced to "s . r".
  The proof proceeds by induction on the structure of $P$.
\end{proof}

\lemref[Lemma]{canonical} is a strong statement about \hyll derivations using
this encoding: every partial derivation using the derived inference rules
represents a prefix of an \spi trace. This is sometimes referred to as
\emph{full adequacy}, to distinguish it from adequacy proofs that require
complete derivations~\cite{nigam08ijcar}. The structure of focused derivations
is crucial because it allows us to close branches early (using init). It is
impossible to perform a similar analysis on unfocused proofs for this encoding;
both the encoding and the framework will need further features to implement a
form of staging~\cite[Chapter 3]{cervesato03tr}.

\begin{corollary}% [soundness]
  \label{thm:adeq} % \mbox{} \newline
  If "env{E}, \cn{rates}, \cinter @ rid ; neg \cact @ rid, can P @ rid
  ; .  ==> . ; (proc Q at rid) tens \cact @ r" is derivable, then "pr E |- pr
  P ->[r]{\!}^* pr Q".
\end{corollary}

\begin{proof}
  Directly from \lemref{canonical}.
\end{proof}

\subsection{Stochastic Correctness with respect to simulation}
\label{sec:simul}

So far the \hyllr encoding of \spi represents any \spi trace symbolically.
However, not every symbolic trace of an \spi process can be produced according to the
operational semantics of \spi, which is traditionally given by a simulator.
%% JD. operational semantics sometimes means SOS, sometimes simulator, in this paper.
This is the main
difference between \hyll (and \spi) and the approach of CSL~\cite{aziz00tcl},
where the truth of a proposition is evaluated against a CTMC, which is why
equivalence in CSL is identical to CTMC bisimulation~\cite{desharmais03jlap}. In
this section we sketch how the execution could be used directly on the canonical
sequents to produce only correct traces (proofs). The proposal in this section
should be seen by analogy to the execution model of \spi simulators such as
SPiM~\cite{phillips04cmmb}, although we do not use the Gillespie algorithm.

The main problem of simulation is determining which of several competing enabled
actions in a canonical sequent to select as the ``next'' action from the
\emph{race condition} of the actions enabled in the sequent. Because of the
focusing restriction, these enabled actions are easy to compute. Each element of
"can P" is of the form "dt -o sum M", so the enabled actions in that element are
given precisely by the topmost occurrences of "neg" in "sum M". Because none of
the sums can have any restricted channels (they have all been removed in the
active decomposition of the process earlier), the rates of all the channels will
be found in the "\cn{rates}" component of the canonical sequent.

The effective rate of a channel "x" is related to its inherent rate by scaling
by a factor proportional to the \emph{activity} on the channel, as defined
in~\cite{phillips04cmmb}. Note that this definition is on the \emph{rate
  constants} of exponential distributions, not the rates themselves. The
distribution of the minimum of a list of random variables with exponential
distribution is itself an exponential distribution whose rate constant is the
sum of those of the individual variables. Each individual transition on a
channel is then weighted by the contribution of its rate to this sum. The choice
of the transition to select is just the ordinary logical non-determinism. Note
that the rounds of the algorithm do not have an associated \emph{delay} element
as in~\cite{phillips04cmmb}; instead, we compute (symbolically) a distribution
over the delays of a sequence of actions.

Because stochastic correctness is not necessary for the main adequacy result in
the previous subsection, we leave the details of simulation to future work.

\section{Related Work}
\label{sec:related}

Logically, the \hyll sequent calculus is a variant of labelled
deduction~\cite{simpson94phd}, a very broad topic not elaborated on here. The
combination of linear logic with labelled deduction isn't new to this work. In
the $\eta$-logic~\cite{deyoung08csf} the constraint domain is intervals of time,
and the rules of the logic generate constraint inequalities as a side-effect;
however its sole aim is the representation of proof-carrying authentication, and
it does not deal with the full generality of constraint domains or with
focusing.
The main feature of $\eta$ not in \hyll is a separate constraint context that
gives new constrained propositions. \hyll is also related to the Hybrid Logical
Framework (HLF)~\cite{reed06hylo} which captures linear logic itself as a
labelled form of intuitionistic logic. Encoding constrained $\pi$ calculi
directly in HLF would be an interesting exercise: we would combine the encoding
of linear logic with the constraints of the process calculus. Because HLF is a
very weak logic with a proof theory based on natural deduction, it is not clear
whether (and in what forms) an adequacy result in \hyll can be transferred to
HLF.

Temporal logics such as CSL and PCTL~\cite{hansson94fac} are popular for logical
reasoning on temporal properties of transition systems with probabilities. In
such logics, truth is defined in terms of correctness with respect to a
constrained forcing relation on the constraint algebra. In CSL and PCTL states
are formal entities (names) labeled with atomic propositions. Formulae are
interpreted on algebraic structures that are discrete (in PCTL) or continuous
(in CSL) time Markov chains. Transitions between states are viewed as pairs of
states labeled with a probability (the probability of the transition), which is
defined as a function from $S \times S$ into $[0,1]$, where $S$ is the set of
states. While such logics have been very successful in practice with efficient
tools, the proof theory of these logics is very complex. Indeed, such modal
logics generally cannot be formulated in the sequent calculus, and therefore
lack cut-elimination and focusing. In contrast, \hyll has a very traditional
proof theoretic pedigree, but lacks such a close correspondence between logical
and algebraic equivalence. Probably the most well known and relevant stochastic
formalism not already discussed is that of stochastic
Petri-nets~\cite{marsan95book}, which have a number of sophisticated model
checking tools, including the PRISM framework~\cite{kwiatkowska04sttt}. Recent
advances in proof theory suggest that the benefits of model checking can be
obtained without sacrificing proofs and proof search~\cite{baelde07cade}.

\section{Conclusion and Future Work}
\label{sec:concl}

We have presented \hyll, a hybrid extension of intuitionistic linear logic with
a simple notion of situated truth, a traditional sequent calculus with
cut-elimination and focusing, and a modular and instantiable constraint system
(set of worlds) that can be directly manipulated using hybrid connectives. We
have proposed two instances of \hyll (i.e two particular instances of the set of
worlds): one modelling temporal constraints and the others modelling stochastic
(continuous time Markov processes) constraints. We have shown how to obtain
representationally adequate encodings of constrained transition systems, such as
the synchronous stochastic $\pi$-calculus in a suitable instance of \hyll.

Several instantiations of \hyll besides the ones in this paper seem
interesting. For example, we can already use disjunction ("plus") to explain
disjunctive states, but it is also possible to obtain a more extensional
branching by treating the worlds as points in an arbitrary partially-ordered set
instead of a monoid. Another possibility is to consider lists of worlds instead
of individual worlds -- this would allow defining periodic availability of a
resource, such as one being produced by an oscillating process. The most
interesting domain is that of discrete probabilities: here the underlying
semantics is given by discrete time Markov chains instead of CTMCs, which are
often better suited for symbolic evaluation~\cite{wu07qest}.

The logic we have provided so far is a logical framework well suited {\it to
  represent} constrained transition systems. The design of a logical framework
{\it for} (i.e. to reason about) constrained transition systems is left for
future work -and might be envisioned by using a two-levels logical framework
such as the Abella system~\cite{gacek09phd}. The work presented in
\cite{deMaria-Despeyroux-Felty:14rep} provides a first step in this direction in
the area of systems biology (where biological systems are viewed as transition
systems), using the Coq proof assistant \cite{BertotCasteran:2004} with
\chyll[T'] (with discrete instants of time) as an object logic. This work can be
seen as a first possible implementation of HyLL with temporal constraints.

An important open question is whether a general logic such as \hyll can serve as
a framework for specialized logics such as CSL and PCTL. A related question is
what benefit linearity truly provides for such logics -- linearity is obviously
crucial for encoding process calculi that are inherently stateful, but CSL
requires no such notion of single consumption of resources.

In the $\kappa$-calculus, reactions in a biological system are
modeled as reductions on graphs with certain state annotations.
It appears (though this has not been formalized)
that the $\kappa$-calculus can be embedded in \hyll even more naturally than
\spi, because a solution---a multiset of chemical products---is simply a tensor
of all the internal states of the binding sites together with the formed
bonds. One important innovation of $\kappa$ is the ability to extract
semantically meaningful ``stories'' from simulations. We believe that \hyll
provides a natural formal language for such stories.

We became interested in the problem of encoding stochastic reasoning in a
resource aware logic because we were looking for the logical essence of
biochemical reactions. What we envision for the domain of ``biological
computation'' is a resource-aware stochastic or probabilistic $\lambda$-calculus
that has \hyll propositions as (behavioral) types.

\vspace{-1em}

\paragraph*{Acknowledgements}
%% \subparagraph*{Acknowledgements}

This work was partially supported by the Information Society Technologies
programme of the European Commission, Future and Emerging Technologies under the
IST-2005-015905 MOBIUS project, and by the European TYPES project.
We thank Fran\c{c}ois Fages, Sylvain Soliman, Alessandra Carbone, Vincent Danos
and Jean Krivine for fruitful discussions on various preliminary versions of the
work presented here. % , in view of its potential applications to biology
Thanks also go to Nicolas Champagnat, Luc Pronzato and Andr{\'e} Hirschowitz
who helped us understand the algebraic nature of stochastic constraints.

\end{document}